\documentclass[envcountsame,oribibl]{llncs}

\usepackage[T1]{fontenc}

\title{A lower bound on opaque sets%
\thanks{%
The work presented here was supported in part by 
JSPS KAKENHI, 
by
the ELC project (Grant-in-Aid for Scientific Research on Innovative Areas, 
MEXT, Japan), 
by 
OTKA under EUROGIGA projects GraDR and ComPoSe 
10-EuroGIGA-OP-003, 
and by 
Swiss National Science Foundation
Grants 200020-144531 and 200021-137574.
}}

\author{Akitoshi Kawamura\inst{1}\and
Sonoko Moriyama\inst{2}\and
Yota Otachi\inst{3}\and
J\'anos Pach\inst{4}}

\institute{University of Tokyo, \email{kawamura@is.s.u-tokyo.ac.jp}\and
Nihon University, \email{moriso@chs.nihon-u.ac.jp}\and 
Japan Advanced Institute of Science and Technology, \email{otachi@jaist.ac.jp}\and 
EPFL, Lausanne and R\'enyi Institute, Budapest, \email{pach@cims.nyu.edu}}

\usepackage[pdftex]{graphicx}
\usepackage[scaled]{helvet}
\usepackage{amsmath,amssymb}

\usepackage{amsthm}

\newcommand{\Rset}{\mathbb R}
\newcommand{\UnitSq}{\square}
\newcommand{\dee}{\mathrm d}
\newcommand{\octagon}{%
\begin{picture}(10,10)(0,1)
  \qbezier(1,1)(5,0)(5,0)
  \qbezier(1,1)(0,5)(0,5)
  \qbezier(1,9)(0,5)(0,5)
  \qbezier(1,9)(5,10)(5,10)
  \qbezier(9,9)(5,10)(5,10)
  \qbezier(9,9)(10,5)(10,5)
  \qbezier(9,1)(10,5)(10,5)
  \qbezier(9,1)(5,0)(5,0)
\end{picture}%
}

\begin{document}

\maketitle

\begin{abstract}
It is proved that the total length of any set 
of countably many rectifiable curves, whose union meets all 
straight lines that intersect the unit square $U$, is at least $2.00002$. 
This is the first improvement on the lower bound of $2$ established by 
Jones in 1964.
A similar bound is proved for all convex sets $U$ other than a triangle.
\end{abstract}

\section{Introduction}

A \emph{barrier} or an \emph{opaque set} 
for $U \subseteq \Rset ^2$ is a set $B \subseteq \Rset ^2$
that intersects every line that intersects $U$. 
For example, when $U$ is a square, 
any of the four sets depicted in Figure~\ref{figure: barriers_square}
is a barrier. 
\begin{figure}[t]
\begin{center}
\includegraphics{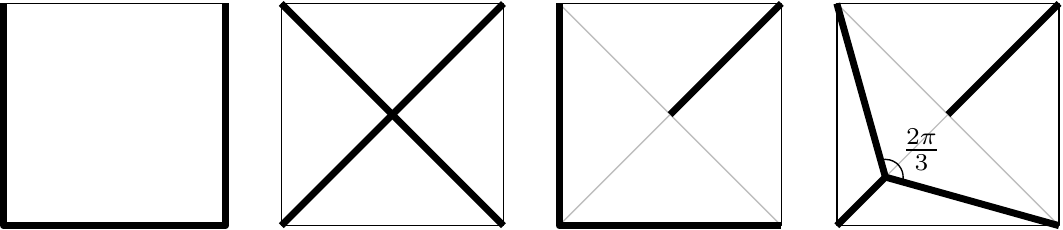}
\caption{Barriers (in thick lines) for the unit square. 
The first one (three sides) and the second one (diagonals) 
have lengths $3$ and $2 \sqrt 2 = 2.828\ldots {}$, respectively.  
The third barrier consists of two sides and half of a diagonal, 
and has length $2 + 1 / \sqrt 2 = 2.707\ldots {}$. 
The last one is the shortest known barrier for the unit square, 
with length $\sqrt 2 + \sqrt 6 / 2 = 2.638\ldots {}$, 
consisting of half a diagonal and the Steiner tree of the lower left triangle. 
}
\label{figure: barriers_square}
\end{center}
\end{figure}
Note that some part of the barrier may lie outside $U$
\begin{figure}
\begin{center}
\includegraphics{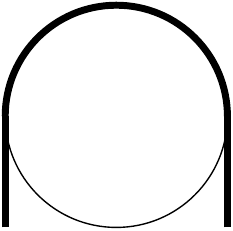}
\caption{A barrier (in thick lines) for a disk
that is shorter than the perimeter.
This is not the shortest one; see \cite{FaberMycielski}.}
\label{figure: barriers_disk}
\end{center}
\end{figure}
(Figure~\ref{figure: barriers_disk}), 
and the barrier need not be connected. 
This notion dates back at least to Mazurkiewicz's work in 1916~%
\cite{Mazurkiewicz}. 

We are interested in ``short'' barriers~$B$ for a given object~$U$, 
and hence we restrict attention to \emph{rectifiable} barriers~$B$. 
By this we mean that $B$ is a union of countably many curves~$\beta$, 
pairwise disjoint except at the endpoints, 
that each have finite length $
\lvert \beta \rvert 
$, and the sum of these lengths converges. 
We call this sum 
the \emph{length} of $B$ and denote it by $\lvert B \rvert$. 

Finding the shortest barrier is hard, 
even for simple shapes~$U$, 
such as the square, the equilateral triangle, and the disk~%
\cite{FaberMycielski, Kawohl}. 
The shortest known barrier for the unit square 
is the rightmost one in Figure~\ref{figure: barriers_square}, 
with length $2.638\ldots {}$. 
This problem and its relatives have an extensive literature. 
See \cite{FaberMycielski, Kawohl2} and the introduction of \cite{DJP10} 
for more history, background, and related problems. 

The best known lower bound for the unit square has been $2$, 
established by Jones in 1964~\cite{Jones}. 
In general, for convex $U$, 
a barrier needs to have length at least 
half the perimeter of $U$
(we review a proof in Section~\ref{section: Jones lower bound}): 

\begin{lemma}
\label{lemma: Jones lower bound}
$\lvert B \rvert \geq p$
for any rectifiable barrier $B$ of 
a convex set $U \subseteq \Rset ^2$ 
with perimeter $2 p$. 
\end{lemma}

Thus, from the point of view of finding short barriers, 
the trivial strategy of 
enclosing the entire perimeter (or the perimeter of the convex hull
if $U$ is a non-convex connected set) gives a $2$-approximation. 
See \cite{DJ13} and references therein
for algorithms that find shorter barriers. 
The current best approximation ratio is 
$1.58\ldots$~\cite{DJP10}. 

Proving a better lower bound has been elusive 
(again, even for specific shapes $U$). 
There has been some partial progress 
under additional assumptions 
about the shape (single arc, connected, etc.)
and location (inside $U$, near $U$, etc.) 
of the barrier~\cite{Croft, DJ14, FaberMycielskiPedersen, Kawohl2, provan}, 
but establishing an unconditional lower bound 
strictly greater than $2$ for the unit square 
has been open (see \cite[Open Problem 5]{DJ13} or \cite[Footnote 1]{DJ14}). 
We prove such a lower bound 
in Section~\ref{section: main}: 

\begin{theorem}
\label{theorem: square}
$\lvert B \rvert \geq 2.00002$ for any rectifiable barrier~$B$ of the unit square~$\UnitSq$. 
\end{theorem}

Dumitrescu and Jiang \cite{DJ14} recently
obtained a lower bound of $2 + 10 ^{-12}$
under the assumption that the barrier lies 
in the square obtained by magnifying $\UnitSq$ by $2$ about its centre. 
Their proof, conceived independently of ours and at about the same time, 
is based on quite different ideas, 
most notably the line-sweeping technique. 
It will be worth exploring whether
their techniques can be combined with ours. 

Our proof can be generalized
(Section~\ref{section: general convex}): 

\begin{theorem}
\label{theorem: general convex}
For any closed convex set $U$ with perimeter~$2 p$
that is not a triangle, 
there is $\varepsilon > 0$ such that 
any barrier~$B$ for $U$ has length at least
$p + \varepsilon$. 
\end{theorem}

Thus, the only convex objects for which 
we fail to establish a lower bound better than Lemma~\ref{lemma: Jones lower bound} 
are triangles. 

The rest of this paper is structured as follows.
In Section~\ref{section: Jones lower bound}, 
we present Jones' proof for Lemma~\ref{lemma: Jones lower bound}. 
We also prove that instead of rectifiable barriers, 
it is sufficient to restrict our attention to 
barriers comprised of line segments.
In Section~\ref{section: lemmas}, 
we present three preliminary lemmas, 
analyzing some 
important special cases in which we can expect to improve on 
Jones' bound. 
The proof of one of these lemmas 
is postponed to Section~\ref{section: proof of lemmas}. 
The three preliminary lemmas are 
combined in Section~\ref{section: main} 
to obtain our lower bound for the length of
a barrier for the square (Theorem~\ref{theorem: square}). 
In Section~\ref{section: general convex}, we show how 
to generalize these arguments to other convex sets 
(Theorem~\ref{theorem: general convex}).
In the last section, we discuss a closely related question.    

\section{Preliminaries: A general lower bound}
\label{section: Jones lower bound}

For a set $U$ and an angle $\alpha \in [0, 2 \pi)$
(all angle calculation will be performed modulo $2 \pi$), 
we write $U (\alpha) \subseteq \Rset$ for the 
image of $U$ projected onto the line 
passing through the origin 
and enclosing angle $+\alpha$ with the positive x-axis, i.e., 
\begin{equation}
 U (\alpha)
= 
 \bigl\{\,
  x \cos \alpha + y \sin \alpha
 :
  (x, y) \in U
 \,\bigr\}
\end{equation}
(Figure~\ref{figure: projection}). 
\begin{figure}[t]
\begin{center}
\includegraphics{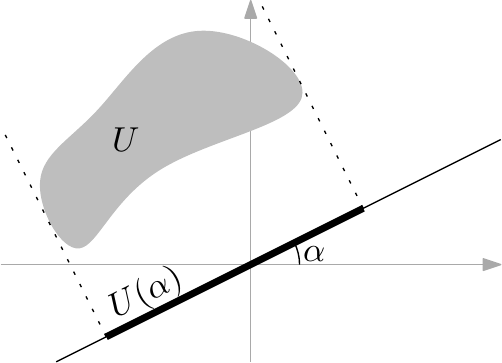}
\caption{The image $U (\alpha) \subseteq \Rset$ of $U$.}
\label{figure: projection}
\end{center}
\end{figure}
To say that $B$ is a barrier of $U$ means that 
$B (\alpha) \supseteq U (\alpha)$ for all $\alpha$. 

For the discussion of upper and lower bounds on the length of a barrier, 
the following lemma says that 
it suffices to consider barriers 
that are a countable union of line segments. 
We call such a barrier \emph{straight}. 

\begin{lemma}[{\cite[Lemma~1]{DJP10}}]
\label{lemma: segment barriers}
Let $B$ be a rectifiable barrier 
for $U \subseteq \Rset ^2$. 
Then, for any $\varepsilon > 0$,
there exists a straight barrier $B _\varepsilon$ for $U$
such that $\lvert B _\varepsilon \rvert \leq (1 + \varepsilon) \lvert B \rvert$.
\end{lemma}

\begin{proof}
Since the proof in \cite{DJP10} has a gap, we provide another proof. 
We will show that for any $\varepsilon > 0$ and any curve $\beta$, 
there is a straight barrier~$\beta''$ of $\beta$ 
of length $\leq (1 + \varepsilon) \lvert \beta \rvert$. 
We can then apply this construction to each curve comprising $B$
to obtain the claimed straight barrier $B _\varepsilon$ of $U$.

If $\beta$ is already a line segment, we are done. 
Otherwise, 
the convex hull~$H$ of~$\beta$ has an interior point. 
Let $\beta'$ be the curve obtained by magnifying $\beta$ 
by $1 + \varepsilon$ about this point. 
Since the convex hull of $\beta'$ 
contains the compact set~$H$ in its interior, 
so does the convex hull of 
a sufficiently fine polygonal approximation $\beta''$ of $\beta'$. 
This implies that $\beta''$ is a barrier of $\beta$. 
\end{proof}

By Lemma~\ref{lemma: segment barriers}, 
we may focus attention on straight barriers: 
$U$ has a rectifiable barrier of length $< l$
if and only if it has a straight barrier of length $< l$. 

As mentioned in the introduction (Lemma~\ref{lemma: Jones lower bound}), 
it has been known that 
any barrier of a convex set must be 
at least half the perimeter. 
We include a short proof of this bound here, 
for completeness and further reference. 
See \cite{cccg08} for another elegant proof. 

\begin{proof}[Proof of Lemma~\ref{lemma: Jones lower bound}]
By Lemma~\ref{lemma: segment barriers}, 
we may assume that $B$ consists of line segments. 
We have
\begin{equation}
\label{equation: width at each angle}
 \lvert 
  U (\alpha)
 \rvert
\leq
 \lvert 
  B (\alpha)
 \rvert
\leq 
 \sum _b 
  \lvert
    b (\alpha)
  \rvert
=
 \sum _b
   \lvert b \rvert 
  \cdot 
   \lvert \cos (\alpha - \theta _b) \rvert 
\end{equation}
for each $\alpha \in [0, 2 \pi)$, 
where the sum is taken over all line segments $b$ that comprise $B$
without overlaps, 
and $\theta _b$ is the angle of $b$. 
Integrating over $[0, 2 \pi)$, we obtain
\begin{equation}
\label{eq: width function integrated}
  \int _{\alpha = 0} ^{2 \pi}
   \lvert
     U (\alpha)
   \rvert
  \, \dee \alpha 
 \leq 
  \sum _{b} 
   \biggl( 
     \lvert b \rvert
    \cdot
     \int _{\alpha = 0} ^{2 \pi}
      \lvert \cos (\alpha - \theta _b) \rvert
     \, \dee \alpha
   \biggr)
 = 
  4 \sum _{b} \lvert b \rvert
 = 
  4 \lvert B \rvert. 
\end{equation}
When $U$ is a convex set, 
the left-hand side equals twice the perimeter. 
\end{proof}

\section{Preliminary lemmas}
\label{section: lemmas}

Note that Theorems \ref{theorem: square} and \ref{theorem: general convex} 
do not merely state
the non-existence of a straight barrier~$B$ of length exactly 
half the perimeter of $U$. 
Such a claim can be proved easily as follows: 
If $B$ is such a barrier, 
the inequality
\eqref{eq: width function integrated}
must hold with equality, 
and so must \eqref{equation: width at each angle}
for each $\alpha$. 
Thus, the second inequality in \eqref{equation: width at each angle} 
must hold with equality, 
which means that $B$ never overlaps with itself
when projected onto the line with angle $\alpha$. 
Since this must be the case for all $\alpha$, 
the entire $B$ must lie on a line, 
which is clearly impossible. 

The theorems claim more strongly 
that a barrier must be longer 
by an absolute constant. 
The following lemma says that in order to obtain such a bound, 
we should find a part $B' \subseteq B$ of the barrier
whose contribution to covering $U$
is less than the optimal by at least a fixed positive constant. 

\begin{lemma}
\label{lemma: waste causes loss}
Let $B$ be a barrier of a convex polygon $U$ of perimeter $2 p$. 
Then 
$\lvert B \rvert \geq p + \delta$ if 
there is a subset $B' \subseteq B$ with
\begin{equation} 
\label{equation: barrier covering nothing}
 \int _{\alpha = 0} ^{2 \pi}
  \lvert B' (\alpha) \cap U (\alpha) \rvert 
 \, \dee \alpha
\leq 
 4 \lvert B' \rvert - 4 \delta. 
\end{equation}
\end{lemma}

\begin{proof}
For each $\alpha \in [0, 2 \pi)$, 
we have $U (\alpha) \subseteq B (\alpha)$, 
and thus 
\begin{align}
  \lvert U (\alpha) \rvert 
&
=
  \lvert B (\alpha) \cap U (\alpha) \rvert 
\leq
  \lvert (B \setminus B') (\alpha) \cap U (\alpha) \rvert 
 +
  \lvert B' (\alpha) \cap U (\alpha) \rvert 
\notag
\\
&
\leq
  \lvert (B \setminus B') (\alpha) \rvert
 +
  \lvert B' (\alpha) \cap U (\alpha) \rvert. 
\end{align}
Integrating over $\alpha \in [0, 2 \pi)$
and using the assumption~\eqref{equation: barrier covering nothing}, 
we get $
 4 p 
\leq
 4 \lvert B \setminus B' \rvert + (4 \lvert B' \rvert - 4 \delta)
=
 4 \lvert B \rvert - 4 \delta$. 
\end{proof}

There are several ways in which 
such a ``waste'' can occur, 
and we make use of two of them
(Figure~\ref{figure: waste}). 
\begin{figure}[t]
\begin{center}
\includegraphics[scale=0.95]{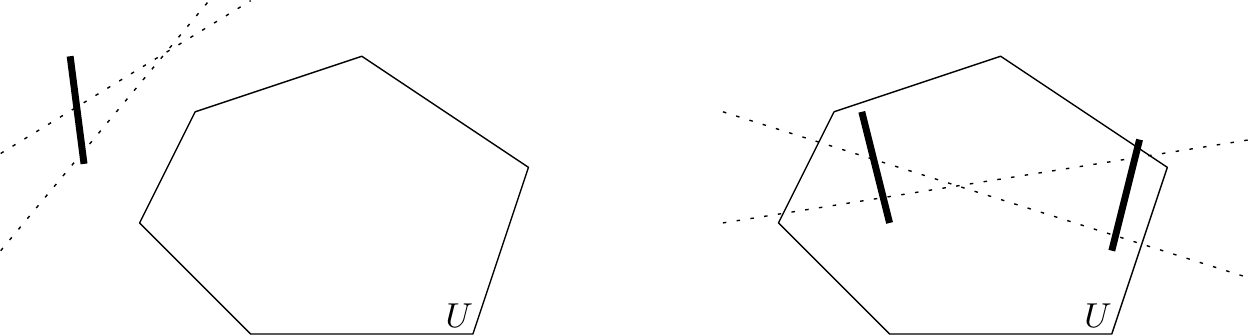}
\caption{Two wasteful situations.
In the left figure, a barrier segment (thick) lies far outside the object~$U$, which leads to significant waste because this segment covers in vain some lines (dotted) that do not pass through $U$; 
this is discussed in Lemma~\ref{lemma: far outside}. 
In the right figure, there are two parts of the barrier (thick) that face each other, which also results in significant waste because they cover some lines (dotted) doubly; 
this is roughly the situation discussed in 
Lemma~\ref{lemma: segment groups}.%
}
\label{figure: waste}
\end{center}
\end{figure}
The first one is when there is a significant part of the barrier 
that lies far outside $U$, 
as described in the following lemma: 

\begin{lemma}
\label{lemma: far outside}
Let $b$ be a line segment that lies outside a convex region $U$. 
Suppose that the set $
A := \{\, \alpha \in [0, 2 \pi) : U (\alpha) \cap b (\alpha) \neq \emptyset \,\}
$ (of angles of all lines through $U$ and $b$)
has measure $\leq 2 \pi - 4 \varepsilon$. 
Then 
\begin{equation}
  \int _{\alpha = 0} ^{2 \pi}
   \lvert b (\alpha) \cap U (\alpha) \rvert 
  \, \dee \alpha
\leq
  4 \lvert b \rvert \cos \varepsilon. 
\end{equation}
\end{lemma}

\begin{proof}
We have 
\begin{equation}
  \int _{\alpha = 0} ^{2 \pi}
   \lvert b (\alpha) \cap U (\alpha) \rvert 
  \, \dee \alpha 
 \leq
  \int _{\alpha \in A}
   \lvert b (\alpha) \rvert 
  \, \dee \alpha 
 =
   \lvert b \rvert
  \cdot
   \int _{\alpha \in A}
    \lvert \cos (\alpha - \theta _b) \rvert 
   \, \dee \alpha 
 \leq
  4 \lvert b \rvert \cos \varepsilon, 
\end{equation}
where the equality in the last inequality is attained when $
 A 
=
  [\varepsilon + \theta _b, \pi - \varepsilon + \theta _b]
 \cup
  [\pi + \varepsilon + \theta _b, 2 \pi - \varepsilon + \theta _b]
$. 
\end{proof}

The second situation where we have 
a significant waste 
required in Lemma~\ref{lemma: waste causes loss} 
is when there are two sets of barrier segments that roughly face each other: 

\begin{lemma}
\label{lemma: segment groups}
Let $\lambda \in (0, \frac \pi 2)$, $\kappa \in (0, \lambda)$ and $l$, $D > 0$. 
Let $B ^-$ and $B ^+$ be unions of $n$ line segments of length~$l$
(Figure~\ref{figure: segment_groups})
such that 
\begin{enumerate}
\item
every segment of $B ^- \cup B ^+$ makes 
angle $> \lambda$ with the horizontal axis; 
\item
$B ^- \cup B ^+$ lies entirely 
in the disk of diameter~$D$ centred at the origin; 
\item
$B ^-$ and $B ^+$ are separated by 
bands of angle $\kappa$ and width $W := n l \sin (\lambda - \kappa)$ centred at the origin, 
as depicted in 
Figure~\ref{figure: segment_groups}---that is, 
each point $(x, y) \in B ^\pm$ satisfies
$
  \pm (x \sin \kappa + y \cos \kappa) \geq W / 2
$ and $
  \pm (x \sin \kappa - y \cos \kappa) \geq W / 2
$ (where $\pm$ should be read consistently as $+$ and $-$). 
\end{enumerate}
\begin{figure}[t]
\begin{center}
\includegraphics{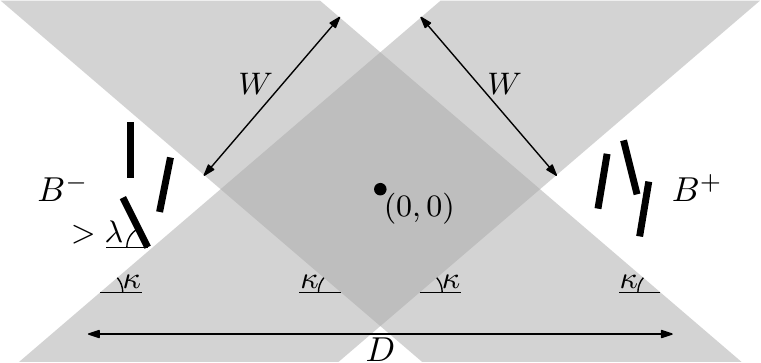}
\caption{Sets $B ^-$ and $B ^+$
(Lemma~\ref{lemma: segment groups}).}
\label{figure: segment_groups}
\end{center}
\end{figure}
Then 
\begin{equation}
\label{equation: parts facing each other}
  \int _{\alpha = 0} ^{2 \pi}
   \lvert (B ^- \cup B ^+) (\alpha) \rvert
  \, \dee \alpha
\leq
  8 n l
 -
  \frac{2 W ^2}{D}. 
\end{equation}
\end{lemma}

Note that $8 n l = 4 \lvert B ^- \cup B ^+ \rvert$, 
so \eqref{equation: parts facing each other} is of the form 
\eqref{equation: barrier covering nothing} in 
Lemma~\ref{equation: barrier covering nothing}. 
The proof of Lemma~\ref{lemma: segment groups}
requires a more involved argument, 
which will be given in Section~\ref{section: proof of lemmas}. 
Before that, 
we prove Theorems \ref{theorem: square} and \ref{theorem: general convex}
using Lemmas \ref{lemma: far outside} and \ref{lemma: segment groups}. 

\section{Proof of Theorem~\ref{theorem: square}}
\label{section: main}

We prove Theorem~\ref{theorem: square}
using Lemmas \ref{lemma: waste causes loss}, 
\ref{lemma: far outside} and \ref{lemma: segment groups}. 
The proof roughly goes as follows. 
Consider a barrier
whose length is very close to $2$. 
\begin{enumerate}
\item 
There cannot be too much of the barrier far outside $\UnitSq$, 
because that would be too wasteful
by Lemma~\ref{lemma: far outside}. 
\item 
This implies that there must be a significant part of the barrier
near each vertex of $\UnitSq$, 
because this is the only place to put barrier segments 
that block those lines intersecting $\UnitSq$ only near this vertex. 
\item 
Among the parts of the barrier that lie near the four vertices, 
there are parts that face each other 
and thus lead to waste by Lemma~\ref{lemma: segment groups}. 
\end{enumerate}

\begin{proof}[Proof of Theorem~\ref{theorem: square}]
Let $\UnitSq$ be the unit square, 
which we assume to be 
closed, 
axis-aligned, 
and centred at the origin. 
Let $B$ be its barrier. 
By Lemma~\ref{lemma: segment barriers}, 
we may assume that $B$ consists of line segments. 
Let $\octagon$ be the octagon (Figure~\ref{figure: octants})
obtained by attaching 
to each edge of $\UnitSq$
an isosceles triangle of height $\frac{29}{590}$ 
(and thus whose identical angles are $\arctan \frac{29}{295}$). 
\begin{figure}[t]
\begin{center}
\includegraphics{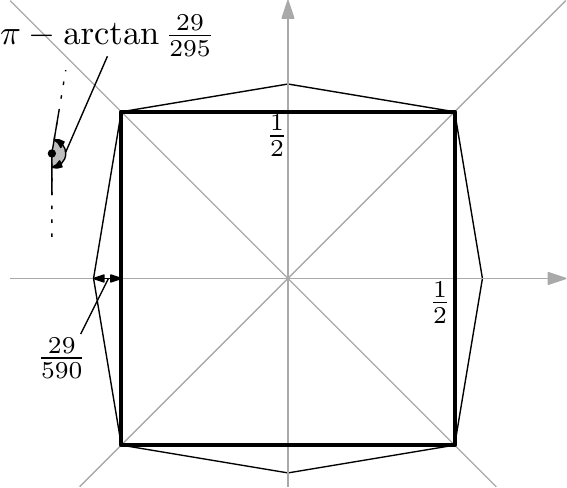}
\caption{Viewed from any point outside the octagon $\protect\octagon$, 
the square $\UnitSq$ lies inside an angle that is smaller than $\pi$ by 
the constant $\arctan \frac{29}{295}$.}
\label{figure: octants}
\end{center}
\end{figure}
Let $B _{\mathrm{out}} = B \setminus \octagon$. 

If $\lvert B _{\mathrm{out}} \rvert > \frac{1}{60}$, 
there is a subset $B' _{\mathrm{out}} \subseteq B _{\mathrm{out}}$
of length $\lvert B' _{\mathrm{out}} \rvert = \frac{1}{60}$
which is a disjoint union of 
finitely many line segments $b$, 
each lying entirely in one of the eight 
regions delimited by the two axes and the two bisectors of the axes. 
Observe that, viewed from each point on $b$, 
the square $\UnitSq$ lies entirely in
an angle measuring $\pi - \arctan \frac{29}{295}$
(Figure~\ref{figure: octants}). 
This allows us to apply Lemma~\ref{lemma: far outside} and obtain
\begin{equation}
  \int _{\alpha = 0} ^{2 \pi}
   \lvert b (\alpha) \cap \UnitSq (\alpha) \rvert 
  \, \dee \alpha
\leq
  4 \lvert b \vert \cos \biggl( \frac 1 2 \arctan \frac{29}{295} \biggr)
<
  4 \lvert b \vert 
 - 
  0.0048 \lvert b \rvert. 
\end{equation}
Summing up for all $b$ 
(and using the triangle inequality), 
we have 
\begin{equation}
  \int _{\alpha = 0} ^{2 \pi}
   \lvert B' _{\mathrm{out}} (\alpha) \cap \UnitSq (\alpha) \rvert 
  \, \dee \alpha
<
  4 \lvert B' _{\mathrm{out}} \vert 
 - 
  0.0048 \lvert B' _{\mathrm{out}} \vert 
=
  4 \lvert B' _{\mathrm{out}} \vert 
 - 
  0.00008, 
\end{equation}
which yields $\lvert B \rvert \geq 2.00002$ 
by Lemma~\ref{lemma: waste causes loss}. 
From now on, we can and will assume that $
\lvert B _{\mathrm{out}} \rvert \leq \frac{1}{60}
$. 

The intersection of $B$ and 
the strip $I _0 := \{\, (x, y) \in \Rset ^2 : \frac 7 8 \leq x + y \leq 1 \,\}$
has length at least $\sqrt 2 / 16$, 
because $
 B (\frac \pi 4) 
\supseteq
 \UnitSq (\frac \pi 4)
=
 [-\sqrt 2 / 2, \sqrt 2 / 2]
\supseteq 
 [\frac 7 8 \sqrt 2 / 2, \sqrt 2 / 2]
$. 
Let $R _0 := I _0 \cap \octagon$ (Figure~\ref{figure: r0})
and $B _0 := B \cap R _0$. 
Then we have
\begin{equation}
 \lvert B _0 \rvert 
=
 \lvert (B \cap I _0) \setminus B _{\mathrm{out}} \rvert 
\geq
 \lvert B \cap I _0 \rvert - \lvert B _{\mathrm{out}} \rvert 
\geq
 \frac{\sqrt 2}{16} - \frac{1}{60}
>
 0.07172
=: 
 2 \eta. 
\end{equation}
Likewise, let $R _1$, $R _2$, $R _3$ be the 
upper left, lower left, and lower right
corners of $\octagon$, respectively.
\begin{figure}[t]
\begin{center}
\includegraphics{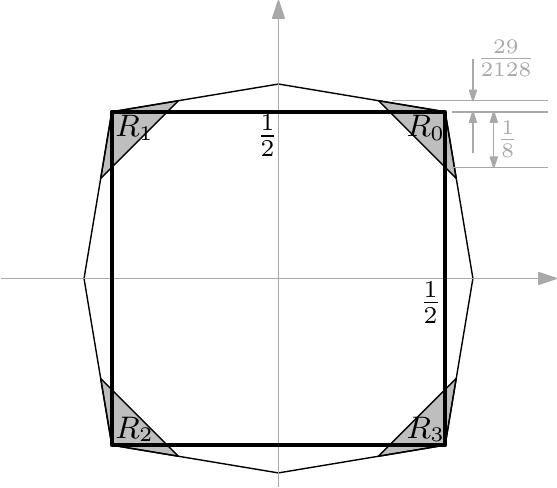}
\caption{The regions $R _0$, $R _1$, $R _2$, $R _3$.}
\label{figure: r0}
\end{center}
\end{figure}
Each of the intersections $B _1$, $B _2$, $B _3$ of $B$ with these regions
has length $> 2 \eta$. 
Observe that 
$R _0 (\frac \pi 2 - 0.1813)$ lies above $R _1 (\frac \pi 2 - 0.1813)$, 
with a gap of size
\begin{equation}
 \frac 7 8 \sin 0.1813 - \biggl( \frac 1 8 + 2 \cdot \frac{29}{2128} \biggr) \cos 0.1813
> 
 0.008; 
\end{equation}
and $R _0 (\frac{3 \pi}{4} - 0.1813)$ lies above $R _2 (\frac{3 \pi}{4} - 0.1813)$, with an even bigger gap. 

For each $i$, we divide $B _i$ into 
three parts $B _{i, i + 1}$, $B _{i, i + 2}$, $B _{i, i + 3}$
(the subscripts are modulo $4$), 
consisting respectively of segments whose angles are in 
$[\frac \pi 2 i - \frac \pi 4, \allowbreak \frac \pi 2 i + \frac \pi 8)$, 
$[\frac \pi 2 i + \frac \pi 8, \frac \pi 2 i + \frac{3 \pi}{8})$ and 
$[\frac \pi 2 i + \frac{3 \pi}{8}, \frac \pi 2 i + \frac{3 \pi}{4})$. 
Thus, 
$B _{i, j}$ consists of segments in $B _i$ that 
``roughly point towards $R _j$.'' 
Since $\lvert B _i \rvert > 2 \eta$, 
we have $\lvert B _i \setminus B _{i, j} \rvert > \eta$
for at least two of the three $j$ for each $i$, 
and thus, for at least eight of the twelve pairs $(i, j)$. 
Hence, there is $(i, j)$ such that $
\lvert B _i \setminus B _{i, j} \rvert > \eta
$ and $
\lvert B _j \setminus B _{j, i} \rvert > \eta
$. 

Let $B ^- \subseteq B _i \setminus B _{i, j}$ 
and $B ^+ \subseteq B _j \setminus B _{j, i}$
be finite unions of line segments of the same length 
such that $\lvert B ^- \rvert = \lvert B ^+ \rvert = \eta$. 
Apply Lemma~\ref{lemma: segment groups} 
to these $B ^-$ and $B ^+$, 
rotated and translated appropriately, 
and the constants 
$\kappa = 0.1813$, 
$\lambda = \frac \pi 8$, 
$D = \sqrt 2$. 
Note that the last assumption of Lemma~\ref{lemma: segment groups} 
is satisfied because $
 W
:=
 \eta \sin (\lambda - \kappa) 
=
 0.03586 \sin (\frac \pi 8 - 0.1813)
=
 0.007524\ldots 
<
 0.008
$. 
This gives
\begin{equation}
  \int _{\alpha = 0} ^{2 \pi}
   \lvert (B ^- \cup B ^+) (\alpha) \rvert
  \, \dee \alpha
\leq
  8 \eta
 -
  \frac{2 W ^2}{D}
<
  8 \eta
 -
  0.00008, 
\end{equation}
whence $\lvert B \rvert \geq 2.00002$ 
by Lemma~\ref{lemma: waste causes loss}. 
\end{proof}

\section{Proof of Theorem~\ref{theorem: general convex}}
\label{section: general convex}

Theorem~\ref{theorem: general convex} is proved by modifying the proof of 
Theorem~\ref{theorem: square} (Section~\ref{section: main}) as follows. 
Let $x _i$ be distinct points ($i = 1$, $2$, $3$, $4$) on the boundary of $U$
at which $U$ is strictly convex, 
i.e., there is a line that intersects $U$ only at $x _i$; 
let $\alpha _i$ be the angle of this line. 
Note that such four points exist unless $U$ is a triangle. 
Let $R _i$ be a sufficiently small closed neighbourhood of $x _i$, 
so that no three of $R _1$, $R _2$, $R _3$, $R _4$ are 
stabbed by a line. 

Instead of the octagon $\octagon$, 
we consider the set $S _\delta \supseteq U$ of points such that 
a random line through this point avoids $U$
with probability less than a positive constant~$\delta$
(Figure~\ref{figure: far_out}). 
\begin{figure}[t]
\begin{center}
\includegraphics{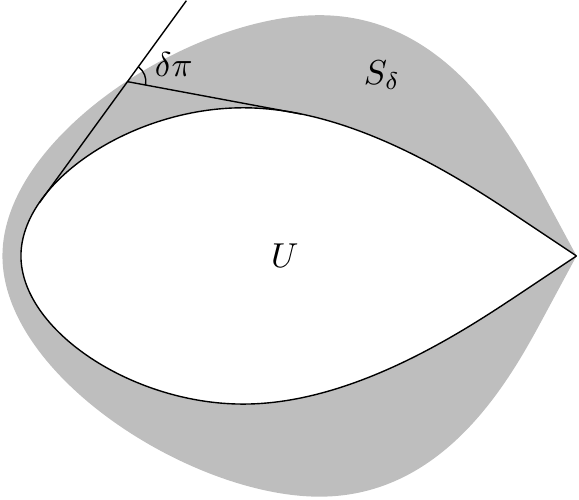}
\caption{$S _\delta$ is the set of points from which $U$ looks big.
Putting too much of the barrier outside $S _\delta$ is wasteful.}
\label{figure: far_out}
\end{center}
\end{figure}
By applying Lemma~\ref{lemma: far outside} in the same way
(with some routine compactness argument), 
we know that $B _{\mathrm{out}} := B \setminus S _\delta$ must be small
(under the assumption of $\lvert B \rvert \leq p + \varepsilon$, 
for an appropriately small $\varepsilon$). 
By taking $\delta$ sufficiently small, 
$S _\delta$ comes so close to $U$ that 
the following happens for each $i = 1$, $2$, $3$, $4$: 
there is a neighbourhood $N \subseteq U$ of $x _i$ in $U$ such that 
every angle-$\alpha _i$ line that intersects $N$ 
intersects $S _\delta$ only in $R _i$. 
This guarantees that the part $B _i := B \cap R _i$ of the barrier 
must have length at least some positive constant
(just to block those angle-$\alpha _i$ lines that hit $N$). 
This allows us to define $B _{i, j}$ in the way similar to 
Theorem~\ref{theorem: square} 
and apply Lemma~\ref{lemma: segment groups} with 
appropriate $\kappa$, $\lambda$, $D$. 

This proves Theorem~\ref{theorem: general convex}. 
To see that the constant~$\varepsilon$ in the statement
must depend on $U$, just consider arbitrarily thin rectangles. 

\section{Proof of Lemma \ref{lemma: segment groups}}
\label{section: proof of lemmas}

It remains to prove Lemma \ref{lemma: segment groups}. 
Let us first interpret 
what it roughly claims. 
By symmetry, 
we can 
halve the interval~$[0, 2 \pi]$
and replace 
\eqref{equation: parts facing each other}
by
\begin{equation}
\label{equation: parts facing each other half}
  4 n l
 -
  \int _{\alpha = 0} ^\pi
   \lvert (B ^- \cup B ^+) (\alpha) \rvert
  \, \dee \alpha
 \geq
  \frac{W ^2}{D}. 
\end{equation}
Let $\mathcal B ^-$ and $\mathcal B ^+$ be the sets of line segments of length~$l$
comprising $B ^-$ and $B ^+$, respectively. 
For each $b \in \mathcal B ^- \cup \mathcal B ^+$, 
consider the region 
\begin{equation}
\label{equation: shadow}
 R _b 
:=
 \{\, (\alpha, v) \in [0, \pi] \times \Rset : v \in b (\alpha) \,\}, 
\end{equation}
whose area is $2 l$. 
Note that the first term 
$4 n l$ 
of \eqref{equation: parts facing each other half}
is the sum of this area for all $b \in \mathcal B ^- \cup \mathcal B ^+$, 
whereas the second term is the area of the union. 
Thus, 
\eqref{equation: parts facing each other half} says that 
the area of the overlap (considering multiplicity) is 
at least $W ^2 / D$. 
Since this term $W ^2 / D$ is proportional to $n ^2$, 
which is the number of pairs $(b, b') \in \mathcal B ^- \times \mathcal B ^+$, 
we should lower-bound (by a constant determined by $\lambda$, $\kappa$, $D$) 
the area of the overlap $R _b \cap R _{b'}$ 
per such pair $(b, b')$. 
This is relatively easy if the overlaps $R _b \cap R _{b'}$ are all disjoint
(using the fact that $R _b$ and $R _{b'}$ must cross roughly in the middle
because of the configuration in Figure~\ref{figure: segment_groups}), 
but it can get tricky otherwise. 

To analyze such a situation, 
we start with the following lemma, 
which makes a similar estimate on the size of potentially complicated overlaps, 
but of simpler objects, namely bands with constant width. 

\begin{lemma}
\label{lemma: crossing}
Let $I \subseteq \Rset$ be an interval and let $W$, $D \geq 0$. 
Let $\mathcal U$ be the set of functions~$f$ 
which take each $\alpha \in I$ to an interval 
$f (\alpha) = [\underline f (\alpha), \overline f (\alpha)]$ of length~$W / n$ 
and are $\frac 1 2 D$-\emph{Lipschitz}, that is, $
 \lvert \underline f (\alpha _0) - \underline f (\alpha _1) \rvert 
\leq
 \frac 1 2 D \cdot \lvert \alpha _0 - \alpha _1 \rvert
$ for each $\alpha _0$, $\alpha _1 \in I$. 
Suppose that $2 n$ functions
$f _1$, \ldots, $f _n$, $g _1$, \ldots, $g _n \in \mathcal U$ satisfy
\begin{align}
\label{equation: much below}
\underline{g _j} (\min I) - \overline{f _i} (\min I) & \geq W, 
&
\underline{f _i} (\max I) - \overline{g _j} (\max I) & \geq W
\end{align}
for each $i$, $j$ 
(i.e., the functions $f _i$ start far below $g _j$ and end up far above). 
Then 
\begin{equation}
\label{equation: overlap}
  \bigl\lvert 
    R _{f _1} \cup \dots \cup R _{f _n}
   \cup
    R _{g _1} \cup \dots \cup R _{g _n}
  \bigr\rvert
\leq
  2 W \lvert I \rvert
 -
  \frac{W ^2}{D}, 
\end{equation}
where $
 R _f
:= 
 \{\, 
  (\alpha, v) \in I \times \Rset 
 :
  v \in f (\alpha) 
 \,\}
$ denotes the graph of $f \in \mathcal U$. 
\end{lemma}

\begin{proof}
Since $\lvert R _f \rvert = \lvert I \rvert \cdot W / n$
for each $f \in \mathcal U$, 
the bound \eqref{equation: overlap} says that 
the loss of area caused by overlaps
is at least $W ^2 / D$. 

We say that $\{f _1, \ldots, f _n\} \subseteq \mathcal U$ is \emph{simple} if 
$R _{f _i} \cap R _{f _j} = \emptyset$ for all distinct $i$, $j$. 
If both $\{f _1, \ldots, f _n\}$ and $\{g _1, \ldots, g _n\}$ are simple
(Figure~\ref{figure: cross_easy}), 
\begin{figure}[t]
\begin{center}
\includegraphics{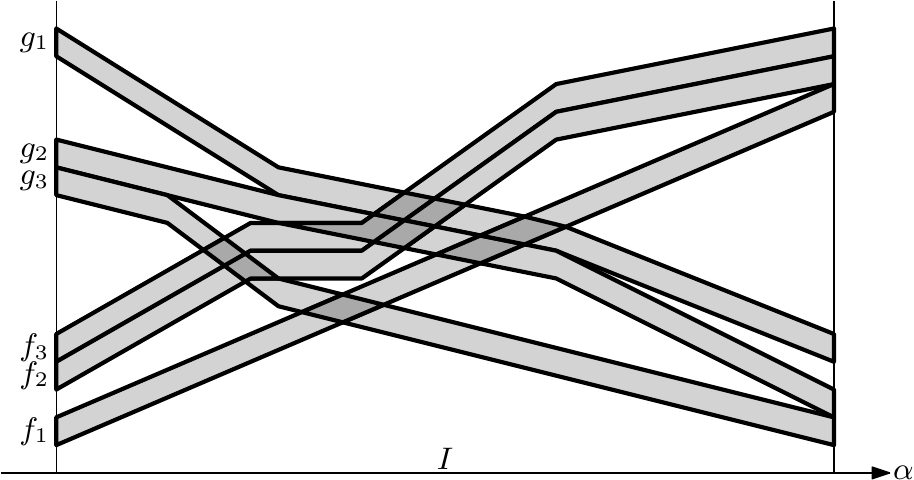}
\caption{Lemma~\ref{lemma: crossing} is easy if $\{f _1, \ldots, f _n\}$ and $\{g _1, \ldots, g _n\}$ are simple. In this case, we only have overlaps of the form $R _{f _i} \cap R _{g _j}$ (nine dark regions in the figure), and we can underestimate their areas separately. We reduce the general case to this easy case.}
\label{figure: cross_easy}
\end{center}
\end{figure}
we easily get \eqref{equation: overlap}, 
because in this case, 
the $n ^2$ overlaps $R _{f _i} \cap R _{g _j}$ are all disjoint, 
and each of them has area $\geq (W / n) ^2 / D$. 
In fact, instead of \eqref{equation: much below}, 
it would have sufficed to assume 
\begin{align}
\label{equation: below}
\underline{g _j} (\min I) \geq \overline{f _i} (\min I), & 
&
\underline{f _i} (\max I) \geq \overline{g _j} (\max I). 
\end{align}

We will reduce the general case 
to this easy special case. 
That is, 
starting from $\{f _1, \ldots, f _n\}$ and $\{g _1, \ldots, g _n\}$ 
that satisfy \eqref{equation: much below}, 
we define simple $\{f' _1, \ldots, f' _n\}$ and $\{g' _1, \ldots, g' _n\}$ 
satisfying \eqref{equation: below}, 
such that 
$R _{f' _1} \cup \dots \cup R _{f' _n} \supseteq R _{f _1} \cup \dots \cup R _{f _n}$ and 
$R _{g' _1} \cup \dots \cup R _{g' _n} \supseteq R _{g _1} \cup \dots \cup R _{g _n}$ 
(note that by these containments, 
the bound \eqref{equation: overlap} for the $f' _i$ and the $g' _j$
implies \eqref{equation: overlap} for the $f _i$ and the $g _j$). 

We describe how we modify $\{f _1, \ldots, f _n\}$ 
to obtain $\{f' _1, \ldots, f' _n\}$ 
(modification on $\{g _1, \ldots, g _n\}$
is done similarly and independently). 
First, 
we make sure that the functions $f _i$ never switch relative positions, 
by exchanging the roles of two intervals 
at every time $\alpha$ at which one overtakes another. 
This way, we ensure that at each time $\alpha$, 
the intervals $f _1 (\alpha)$, \ldots, $f _n (\alpha)$ are in ascending order. 

These intervals may still overlap one another. 
So we push them upwards one by one as necessary 
to avoid previous intervals. 
Specifically, define $f' _1$ by $f' _1 (\alpha) = f _1 (\alpha)$, 
and then $f' _i$ for $i = 2$, $3$, \ldots by
$\underline{f' _i} (\alpha) = \max \{\underline{f _i} (\alpha), \max \overline{f' _{i - 1}} (\alpha) \}$. 

The functions remain $\frac 1 2 D$-Lipschitz, 
since the resulting $f' _i$ 
at each time $\alpha$
moves at the same speed as one of the original $f _j$. 
The condition \eqref{equation: below} is also satisfied, 
since initially we had \eqref{equation: much below} 
and then we moved each $f' _i$ upwards by at most $W$. 
\end{proof}

\begin{proof}[Proof of Lemma~\ref{lemma: segment groups}]
Let $\mathcal B ^-$, $\mathcal B ^+$ and $R _b$
be as at the beginning of Section~\ref{section: proof of lemmas}. 
As explained there, 
our goal it to prove 
\eqref{equation: parts facing each other half}, 
which says that 
the area of the overlap 
between $R _b$ for $b \in \mathcal B ^- \cup \mathcal B ^+$ 
is at least $W ^2 / D$. 
We claim that this is true 
\emph{even if we replace each $R _b$ by its subset $\tilde R _b$
defined below}. 

Let $I := [\frac \pi 2 - \kappa, \frac \pi 2 + \kappa]$. 
Note that, because of the configuration of segments (Figure~\ref{figure: segment_groups}), 
we have $\lvert b (\alpha) \rvert \geq l \sin (\lambda - \kappa) = W / n$
for each $\alpha \in I$ and $b \in \mathcal B ^- \cup \mathcal B ^+$. 
We define a subset of $R _b$ (see \eqref{equation: shadow})
by restricting $\alpha$ to $I$ and replacing the interval $b (\alpha)$
by its subinterval $
 \tilde b (\alpha) 
:=
 [\min (b (\alpha)), \min (b (\alpha)) + W / n]
$: 
\begin{equation}
 \tilde R _b 
:=
 \{\, (\alpha, v) \in I \times \Rset : v \in \tilde b (\alpha) \,\}. 
\end{equation}
Our claim was that 
the total area of pairwise overlaps between $\tilde R _b$
for $b \in \mathcal B ^- \cup \mathcal B ^+$ is at least $W ^2 / D$. 
But this is 
Lemma~\ref{lemma: crossing}
applied to 
$\{f _1, \dots, f _n\} = \{\, \tilde b : b \in \mathcal B ^- \,\}$, 
$\{g _1, \dots, g _n\} = \{\, \tilde b : b \in \mathcal B ^+ \,\}$. 
\end{proof}

\section{Half-line barriers}

Let us finally propose an analogous question, 
obtained by replacing lines by half-lines
in the definition of barriers: 
a set $B \subseteq \Rset ^2$ 
is a \emph{half-line barrier} of
$U \subseteq \Rset ^2$
if all half-lines intersecting $U$ intersect $B$. 
This intuitively means 
``hiding the object~$U$ from outside,'' 
which we find perhaps as natural, if not more, than the notion of opaque sets.
Similarly to Lemma~\ref{lemma: Jones lower bound}, we have

\begin{lemma}
$\lvert B \rvert \geq p$ for any 
rectifiable half-line barrier~$B$
of a convex set $U \subseteq \Rset ^2$ with perimeter $p$. 
\end{lemma}

Thus, unlike for line barriers, 
the question is completely answered when $U$ is connected: 
the shortest half-line barrier 
is the boundary of the convex hull. 

If $U$ is disconnected, 
there can be shorter half-line barriers. 
For example, if $U$ consists of 
two connected components that are enough far apart from each other, 
it is more efficient to cover them separately 
than together. 
One might hope that an optimal half-line barrier
is always obtained by grouping the connected components of $U$ in some way and 
taking convex hulls of each. 
This is not true, as the example in Figure~\ref{figure: disjoint} shows.
\begin{figure}[t]
\begin{center}
\includegraphics[scale=0.9]{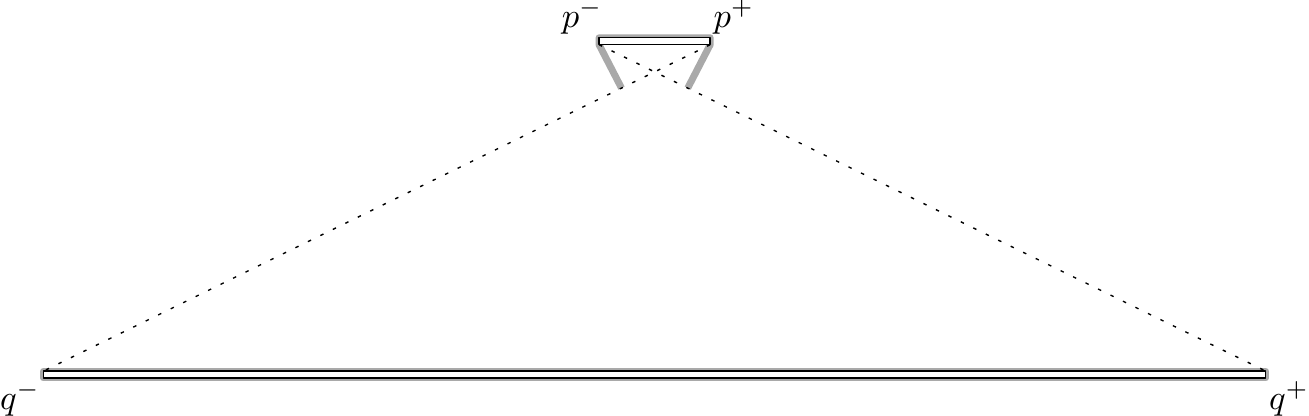}
\caption{Consider the line segments $p ^- p ^+$ and $q ^- q ^+$, 
where $p ^\pm = (\pm 1, 8)$ and $q ^\pm = (\pm 15, 0)$, 
and let $U$ 
be the union of these segments with small ``thickness'': 
$U$ consists of 
a rectangle with vertices $(\pm 1, 8 \pm \varepsilon)$ 
and another with vertices $(\pm 15, \pm \varepsilon)$, 
for a small $\varepsilon > 0$. 
The boundaries of these thick line segments 
have total length $64$ (plus a small amount due to the thickness). 
The boundary of the convex hull of all of $U$ has length
$2 + 30 + 2 \sqrt{260} > 64.24$ (plus thickness). 
But we have another half-line barrier
depicted above in gray, 
whose total length is $
 2 + 60 + 2 / \sqrt 5 + 2 / \sqrt 5 
< 
 63.79
$ (plus thickness, which can be made arbitrarily small).}
\label{figure: disjoint}
\end{center}
\end{figure}
We have not been able to find an algorithm 
that achieves a nontrivial approximation ratio for this problem. 

\subsubsection*{Acknowledgements.}
We are grateful to G\'abor Tardos for many interesting discussions on the subject. In particular, the present proof of Lemma~\ref{lemma: segment barriers} is based on his idea.

\end{document}